\documentclass{article}

\usepackage{arxiv}

\usepackage[utf8]{inputenc} 
\usepackage[T1]{fontenc}    
\usepackage{hyperref}       
\usepackage{url}            
\usepackage{booktabs}       
\usepackage{amsfonts}       
\usepackage{nicefrac}       
\usepackage{microtype}      
\usepackage{lipsum}		
\usepackage{graphicx}
\usepackage{doi}

\usepackage{todonotes}
\usepackage{dsfont}
\usepackage{mathtools}
\usepackage{amsmath}

\DeclareMathOperator{\cost}{Cost}

\DeclareMathOperator{\spl}{Split}
\DeclareMathOperator{\merge}{Merge}
\DeclareMathOperator{\move}{Move}

\usepackage{amsthm}

\newtheorem{theorem}{Theorem}[section]
\newtheorem{lemma}[theorem]{Lemma}

\title{Exact and Heuristic Approaches to Speeding Up the MSM Time Series Distance Computation}

\author{ \href{https://orcid.org/0000-0003-4463-1149}{\includegraphics[scale=0.06]{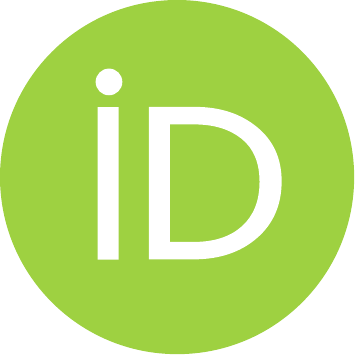}\hspace{1mm}Jana Holznigenkemper}\\
	Department of Mathematics and Computer Science\\
	University of Marburg\\
	Germany \\
	\texttt{holznigenkemper@mathematik.uni-marburg.de} \\
	\And
	\href{https://orcid.org/0000-0003-0829-7032}{\includegraphics[scale=0.06]{orcid.pdf}\hspace{1mm}Christian Komusiewicz} \\
	Department of Mathematics and Computer Science\\
	University of Marburg\\
	Germany \\
	\texttt{komusiewicz@mathematik.uni-marburg.de} \\
    \And
	\href{https://orcid.org/0000-0002-9362-153X}{\includegraphics[scale=0.06]{orcid.pdf}\hspace{1mm}Bernhard Seeger} \\
	Department of Mathematics and Computer Science\\
	University of Marburg\\
	Germany \\
	\texttt{seeger@mathematik.uni-marburg.de} \\
}

\renewcommand{\headeright}{Technical Report}
\renewcommand{\undertitle}{Technical Report}
\renewcommand{\shorttitle}{Speeding Up the MSM Time Series Distance Computation}

\hypersetup{
pdftitle={Exact and Heuristic Approaches to Speeding Up the MSM Time Series Distance Computation},
pdfsubject={q-bio.NC, q-bio.QM},
pdfauthor={Jana Holznigenkemper, Christian Komusiewicz, Bernhard Seeger},
pdfkeywords={Time Series Distance Computation, Dynamic Programming, Time Series Metric, Heuristics},
}

\begin{document}
\maketitle

\begin{abstract}
The computation of the distance of two time series is time-consuming for any elastic distance function that accounts for misalignments. Among those functions, DTW is the most prominent. However, a recent extensive evaluation has shown that the move-split merge (MSM) metric is superior to DTW regarding the analytical accuracy of the 1-NN classifier. Unfortunately, the running time of the standard dynamic programming algorithm for MSM distance computation is~$\Omega(n^2)$, where $n$ is the length of the longest time series. In this paper, we provide approaches to reducing the cost of MSM distance computations by using lower and upper bounds for early pruning paths in the underlying dynamic programming table. For the case of one time series being a constant, we present a linear-time algorithm. In addition, we propose new linear-time heuristics and adapt heuristics known from DTW to computing the MSM distance. One heuristic employs the metric property of MSM and the previously introduced linear-time algorithm. Our experimental studies demonstrate substantial speed-ups in our approaches compared to previous MSM algorithms. In particular, the running time for MSM is faster than a state-of-the-art DTW distance computation for a majority of the popular UCR data sets.
\end{abstract}

\keywords{Time Series Distance Computation \and Dynamic Programming \and Time Series Metric \and Heuristics}

\section{Introduction}
Measuring the distance between two time series is a crucial step in time series analysis tasks such as classification~\cite{AbandaML19}.
An important class of time series distance measures are elastic measures like Dynamic Time Warping (\textsc{DTW})\cite{BerndtC94}, the Move-Split-Merge (\textsc{MSM}) metric~\cite{StefanAD13}, Longest Common Subsequence (\textsc{LCS}) based distance~\cite{VlachosGK02}, and the Edit Distance with Penalty (\textsc{EDP})\cite{ChenN04}.
All these measures have in common that their computation relies on dynamic programming. 
Among these distance measures, \textsc{DTW} is the most frequently used. 
In a recent study, Paparrizos et~al.~\cite{PaparrizosLEF20} re-examined 71 time series distance functions for 1-NN classification task and found \textsc{MSM} to lead to high-accuracy results in classification tasks.
In addition to this, MSM offers one significant further advantage:
it satisfies the properties of a mathematical metric, which is advantageous for example in database indexing. 
Paparrizos et.~al~\cite{PaparrizosLEF20} however also identified a key drawback of the \textsc{MSM} distance: The state-of-the-art \textsc{MSM} algorithm with quadratic complexity was shown to be significantly slower than \textsc{DTW} (with band heuristic)~\cite{PaparrizosLEF20}. 
In this paper, we aim to remedy this drawback. To this end, we examine several strategies to  speed up the computation time, based on the computation of heuristic lower and upper bounds and their use in an exact algorithm. 
More precisely, our contributions are: 
\begin{itemize}
    \item We propose a linear-time algorithm computing the exact distance between a constant time series and an arbitrary time series. 
    \item We adapt common heuristic strategies like the Sakoe-Chiba band and the Itakura parallelogram for \textsc{MSM} to reduce the amount of entries to be computed in the dynamic programming table.
    \item We develop new \textsc{MSM}-specific heuristics making use of the transformation structure and the triangle inequality. 
    \item We speed up the exact \textsc{MSM} distance computation by introducing \textit{PrunedMSM}, an exact algorithm employing pruning strategies using several lower and upper bounds. 
    \item In experiments on samples of real-world time series, we show a substantial running time advantage of \textit{PrunedMSM} over the classic \textsc{MSM} implementation and analyze the trade-off between running time and accuracy for all given heuristics. \item We compare the fastest variant of \textit{PrunedMSM} to a state-of-the-art \textsc{DTW} distance computation showing that \textit{PrunedMSM} is faster than \textit{PrunedDTW} in most cases. 
\end{itemize}
The remainder of the paper is structured as follows. Section~\ref{sec:RelatedWork} reviews related work. In Section~\ref{chap:preliminaries}, we recall the definition of the MSM metric.
Then, in Section~\ref{sec:Speedup}, we discuss first simple speed-ups to improve the running time of the classic \textsc{MSM} implementation and give an exact linear-time algorithm for the computation of the distance between an arbitrary and a constant time series. 
Various heuristic strategies are given in Section~\ref{sec:Heuristics}.
The new exact \textit{PrunedMSM} algorithm is presented in Section~\ref{sec:pruningStrategies}, including several lower and upper bounding strategies. 
We experimentally discuss all given heuristics, evaluate \textit{PrunedMSM}, compare the fastest exact \textit{PrunedMSM} to \textit{PrunedDTW} in Section~\ref{sec:experiments}, and conclude in Section \ref{sec:Conclusion}. 

\section{Related Work}
\label{sec:RelatedWork}

Measuring the distance between two time series is crucial for many data mining applications like clustering~\cite{aghabozorgi2015time}, classification~\cite{AbandaML19}, or motif discovery~\cite{torkamani2017survey}. 
The most common and simple distance linear-time measure is the Euclidean distance (ED).
ED was shown to be inferior in classification tasks since it is sensitive to distortion of the time axis~\cite{YiJF98}. Furthermore, ED cannot measure the distance of time series of different lengths. 
Another class of time series distances offer elastic measures like \textsc{DTW}~\cite{BerndtC94}, \textsc{MSM}~\cite{StefanAD13}, \textsc{LCS}~\cite{VlachosGK02}, and \textsc{EDP}~\cite{ChenN04}.
The computation of all these measures relies on dynamic programming and a two-dimensional dynamic programming table with quadratic complexity.
To reduce the running time, heuristic strategies are used to decrease the running time. 
A common technique is to reduce the number of entries to be computed in the dynamic programming table by introducing a global constraint. For example, the Sakoe-Chiba band~\cite{sakoe1978dynamic} and the Itakura parallelogram~\cite{itakura1975minimum} are two common approaches for limiting the space to the left and to the right of the (slanted) diagonal of the table. 
The Sakoe-Chiba band has been evaluated for classification tasks, e.g., for \textsc{DTW}, \textsc{LCS}~\cite{VlachosGK02}, and \textsc{EDP}~\cite{kurbalija2014influence}.
The results show that the constrained versions qualitatively differ from their unconstrained ones regarding the classification error rates. 
Further studies that analyze the application of the Sakoe-Chiba band to \textsc{DTW} have shown that the classification of DTW with band is more accurate than the one without~\cite{RatanamahatanaK05,geler2019dynamic}. 
A comparison between the Sakoe-Chiba band for different band sizes and the Itakura parallelogram show a higher accuracy for the Sakoe-Chiba band~\cite{geler2019dynamic}. We are not aware of any studies on how the Itakura parallelogram performs with different scales. Moreover, we are not aware that these two heuristics have been used for the \textsc{MSM} metric. 

There are other heuristics for \textsc{DTW}, like~FastDTW~\cite{salvador2007toward}, LuckyDTW~\cite{SpiegelJA14}, and AnytimeDTW~\cite{ZhuBRK12}.
FastDTW finds a warping path for low resolutions of the time series. The warping path is then projected to a higher resolution until the final resolution is reached. 
LuckyDTW greedily determines a warping path. 
AnytimeDTW is an anytime algorithm where you can get the best-so-far result during the computation. 
None of these approaches give an accurate approximation factor of how far the solution differs from the exact distance.
Silva and Batista~\cite{SilvaB16} developed the exact \textsc{DTW} distance computation \textsc{PrunedDTW} with better running time. The idea is to prune the table entries that are guaranteed not to be part of the warping path. 
\textsc{PrunedDTW} is used in the UCR suite \cite{SilvaGKB18} to further accelerate the running time for time series similarity search. 

We are not aware of any work that considers pruning techniques for \textsc{MSM} computation.

\section{Preliminaries}
\label{chap:preliminaries}

For~$k\in\mathds{N}$, let $[k]:= \{1,\ldots ,k\}$. 
A \emph{time series} of length~$m$ is a sequence~$x=(x_1, \ldots , x_m)$, where each \emph{data point}, in short \emph{point},  $x_i$ is a real number.

\subsection{Move-Split-Merge Operations}
We now define the \textsc{MSM} metric, following the notation of Stefan et al.~\cite{StefanAD13}. 
The \textsc{MSM} metric allows three transformation operations to transfer one time series into another: 
\emph{move, split}, and \emph{merge} operations. 
For a time series $x=(x_1, \ldots , x_m)$ a move transforms a point~$x_i$ into $x_i + w$ for some~$w\in\mathds{R}$. More precisely,  $\move_{i,w}(x)~:=~(x_1, \ldots , x_{i-1},x_i~+~w, x_{i+1}, \ldots ,x_m) $, with cost $\cost(\move_{i,w}) = \vert w \vert$.
We say that there is a \emph{move at point $x_i$ to another point}~$x_i + w$. The split operation splits some  element of $x$ into two consecutive points. 
Formally, a split at point $x_i$ is defined as $\spl_i(x):= (x_1, \ldots , x_{i-1}, x_i, x_i, x_{i+1}, \ldots, x_m).$ 
A merge operation may be applied to two consecutive points of equal value. 
For $x_i = x_{i+1}$, it is given by $\merge_i(x)~\coloneqq~(x_1, \ldots , x_{i-1}, x_{i+1}, \ldots, x_m).$
We say that~$x_i$ and $x_{i+1}$ \emph{merge}.
Split and merge operations are inverse operations with equal cost that is determined by a given nonnegative constant $c  = \cost(\spl_{i}) =  \cost(\merge_{i})$. 
A \emph{sequence of transformation operations} is a tuple  $\mathds{S} = (S_1, \ldots , S_s)$, where $S_j \in \{\move_{i_j,w_j}, \spl_{i_j}, \merge_{i_j}\}$.
A \emph{transformation} $T(x, \mathds{S})$ of a time series $x$ by a sequence of transformation operations $\mathds{S}$ is defined as $T(x, \mathds{S}) \coloneqq T(S_1(x), (S_2, \ldots , S_s))$.
If $\mathds{S}$ is empty, we define $T(x, \emptyset) \coloneqq x$. The cost of a sequence of transformation operations $\mathds{S}$ is the sum of all individual operations cost, that is, $ \cost(\mathds{S}) \coloneqq \sum_{S \in \mathds{S}}\cost(S).$
We say that $\mathds{S}$ \emph{transforms~$x$ to~$y$} if $T(x,\mathds{S}) = y$.
A transformation is \emph{optimal} if it has minimal cost transforming $x$ to $y$.
The \emph{\textsc{MSM} distance} $d(x,y)$ between two time series~$x$ and~$y$ is the cost of an optimal transformation.

\subsection{Transformation Graphs}
We briefly introduce the concept of \emph{transformation graphs} to describe the structure of a transformation $T(x,\mathds{S}) = y$. For
more detailed information see Appendix~\ref{app:propertiesTrafoGraphs} and the work of Holznigenkemper et al.~\cite{holznigenkemper2023computing} and Stefan et al.~\cite{StefanAD13}. The transformation $T(x,\mathds{S}) = y$ can be described by a directed acyclic graph $G_{\mathds{S}}(x,y)$, the \emph{transformation graph}.
The edges represent the transformation operations of $\mathds{S}$. 
To create a transformation graph, for each operation in $\mathds{S}$, a \emph{move edge}, or two \emph{split} or \emph{merge edges} are added to the graph. 
An example is depicted in Figure~\ref{fig:example_transformation_graph}. 
A \emph{transformation path} in~$G_{\mathds{S}}(x,y)$ is a directed path from a source node~$x_i$ to a sink node $y_j$,
we say that~$x_i$ is \emph{aligned} to~$y_j$. 

\begin{figure}[t]
    \centering
    \includegraphics[width = 0.2\textwidth]{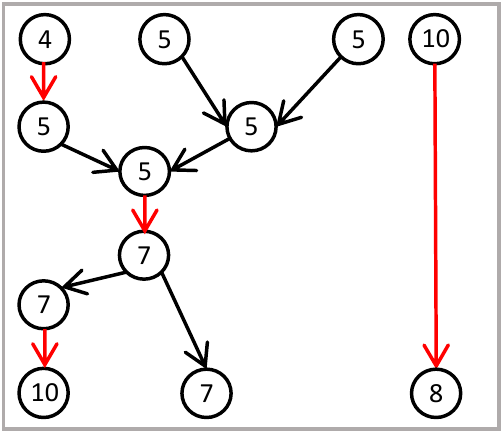}
    \caption{Optimal transformation graph of $x=(4,5,5,10)$ to $y=(10,7,8)$ for $c=0.1$. Move edges are red. The cost of a move edge is the difference between the source and the target point. We have total cost merge and split cost $3c$ and move cost of 8. Hence, the distance between $x$ and $y$ is $d(x,y) = 8.3$.}
    \label{fig:example_transformation_graph}
\end{figure}

\section{Speeding up Exact Distance Computations}
\label{sec:Speedup}
We first sketch the original dynamic programming algorithm of the \textsc{MSM} distance and  
 discuss common techniques for speeding up this algorithm. Second, we give a linear-time algorithm for the exact distance computation of an arbitrary and a constant time series. 

\subsection{Speeding up Classic \textsc{MSM}}
Stefan et. al~\cite{StefanAD13} give the following dynamic programming algorithm for computing the \textsc{MSM} metric on two input time series $x=(x_1,\ldots, x_m)$ and $y=(y_1,\ldots,  y_n)$. 
The algorithm fills a two-dimensional table~$D$ where
an entry~$D[i,j]$ represents the cost of transforming the partial time series $(x_1, \ldots, x_i)$ to the partial time series $(y_1,\ldots, y_j)$. The distance $d(x,y)$ is given by $D[m,n]$.
The recursive formulation of the \textsc{MSM} metric returns the minimum of the cost for the three transformation operations.
$$
D[i,j] = \min\{A_{MO}[i,j], A_{M}[i,j], A_{SP}[i,j] \}, \text{where} 
$$
\begin{align*}
A_{MO}[i,j] &= D[i-1,j-1] + \vert x_i -y_i \vert &(move) \\
A_{M}[i,j] &= D[i-1,j] + C(x_i, x_{i-1}, y_j) &(merge) \\
A_{SP}[i,j] &= D[i,j-1] + C(y_j, x_i, y_{j-1}) &(split) 
\end{align*}
$$
C(x_i,x_{i-1},y_j) = \begin{cases} 
 c, \text{if } x_{i-1} \leq x_i \leq y_j \text{ or } x_{i-1}\geq x_i \geq y_j&\\
 c + \min(\vert x_i-x_{i-1}\vert,\vert x_i-y_j\vert),  \text{   else.}&
\end{cases}
$$
Note that there are special cases for computing the first column and first row of $D^*$. Then, it holds that $D[i,1] = D[i-1,1] + C(x_i, x_{i-1}, y_1) $ (only merge operation may be further applied) and $D[1,j] =D[1,j-1] + C(y_j, x_1, y_{j-1})$ (only split operation may be further applied). 
For the base case, only a move operation is allowed, that is, $D[1,1] = \vert x_1 -y_1\vert$. 
 
To simplify the computation rules in our implementation, we add an additional row and column to the table with $D[0,0]=0$ and $D[i,0]= D[0,j]= \infty$. The computation starts at $D[1,1]$ without the need of treatments of special cases anymore. Second, we adopt a common strategy  to speed up the running time via optimizing the storage usage. Rather than using space for a two-dimensional array of size ($(m+1)\times (n+1)$), we only use a one-dimensional array of size $(n+1)$ and compute the table row by row assuming that $n\leq m$. This also reduces the computational overhead due to a better cache locality.  
These speed-up techniques are applied to all further approaches where a dynamic programming table is used.

\subsection{\textsc{cMSM}}
The following lemma states that the distance between a time series $x$ and a constant time series can be computed in linear time. 
\begin{lemma}
Given a time series $x=(x_1,\ldots,x_m)$ and a constant time series $q^{(m)}=(q_1,\ldots, q_m)$ with $q_i = q$ $\forall i\in [m]$, the distance $d(x,q^{(m)})$ can be computed in linear time. 
\end{lemma}
The \textsc{cMSM} algorithm is based on a decomposition of the transformation graph $G_{\mathds{S}}(x,q^{(m)})$. 
A formal proof of its correctness can be found in Appendix~\ref{app:cMSM}.
In the following, we just give the rules for the dynamic program computing $d(x,q^{(m)})$. 
A one-dimensional table $D_c$ is filled in reverse order where an entry $D_c[i]$ represents the cost transforming the partial time series $(x_i,\ldots, x_m)$ to $(q_i,\ldots,q_m)$.
The idea of the computation is that if the Euclidean distances between $q$ and consecutive points~$x_i$ and~$x_{i+1}$ are both greater than $2c$, then $x_i$ and  $x_{i+1}$ are merged. Otherwise, there is a move from $x_i$ to~$q$. 
Figure \ref{fig:constantDist} depicts the logic of the cost computation. 
The base case is always a move operation, i.e., $D_c[m] = \vert x_m -q \vert$. 
For all other entries, we check the above-described condition. 
We can compute the cost per point. Formally, if $\vert x_i-q \vert \geq 2c$ and $\vert x_{i+1}-q \vert \geq 2c$ then $D_c[i] = D_c[i+1] +2c + \max(0, \vert x_{i+1} - q \vert - \vert x_i -q \vert)$,
otherwise $D_c[i] = D_c[i+1] + \vert x_i - q \vert $. 
\begin{figure}
    \centering
    \includegraphics[width=0.45\textwidth]{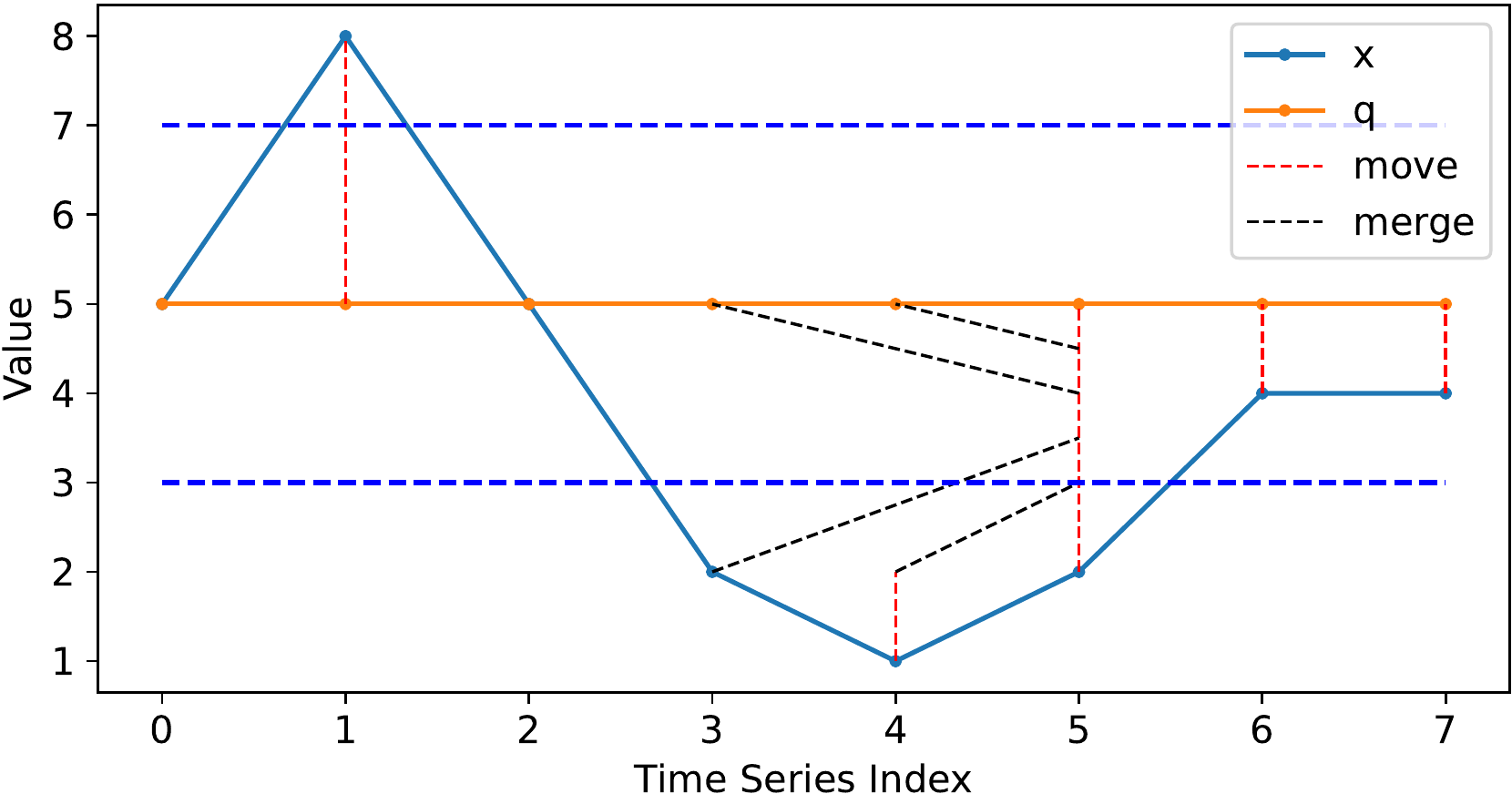}
    \caption{Alignment of $x=(5,8,5,2,1,2,4,4)$ to a constant time series with $q=5$. All merge operations are marked in red, merges and splits are marked in black. The two dotted lines refer to $c$ with $c=1$. The resulting dynamic programming table is $D_c=[13,13,10,10,8,5,2,1]$ with $d(x,q^{(8)})=13$.}
    \label{fig:constantDist}
\end{figure}
Computing the whole distance, there is no difference in filling the table in the right or reverse order. 
The reverse filling is only important for computing an upper and lower bound as discussed in Sections~\ref{sec:TriangleHeuristic} and~\ref{sec:LBtriangle}. 
We use the \textsc{cMSM} algorithm in the next section to develop a heuristic computing an approximate \textsc{MSM} distance. 

\section{Heuristics}
\label{sec:Heuristics}
In this section we develop different heuristics to obtain an approximation of the \textsc{MSM} distance.
The first approach, the \emph{triangle heuristic} results from applying \textsc{cMSM} and the upper triangle inequality. 
Second, the \emph{greedy heuristic} takes advantage of the particular transformation operations of \textsc{MSM}. 
Finally, we apply two state-of-the-art techniques to reduce the number of entries that are computed in the DP table: The Sakoe-Chiba band~\cite{sakoe1978dynamic} and the Itakura parallelogram~\cite{itakura1975minimum}.  

\subsection{Triangle Heuristic}
\label{sec:TriangleHeuristic}
Given two time series $x=(x_1,\ldots,x_m)$ and $y=(y_1,\ldots,y_n)$ we apply \textsc{cMSM} and the upper triangle inequality to compute $d_{triangle}(x,y)$. 
Without loss of generality, let $m \ge n$ and $q^{(m)}$ a constant time series of length $m$.  
By the triangle inequality $d(x,y) \leq d(x,q^{(m)}) + d(y,q^{(m)})$. 
\textsc{cMSM} is only applicable for time series of equal lengths, that means we can not compute $d(y,q^{(m)})$ directly with this algorithm. 
We compute $d(y,q^{(n)})$ with \textsc{cMSM} and add the split cost to align $y_n$ to $(q_{n+1},\ldots q_m)$. 
This leads to the following heuristic: 
\begin{equation*}
d_{triangle}(x,y) = d(x,q^{(m)}) + d(y,q^{(n)}) + (m-n)\cdot c.
\end{equation*}

\subsection{Greedy Heuristic} 
Computing the greedy heuristic $d_{greedy}(x,y)$ of two time series $x$ and $y$ follows a similar logic as the computation of \textsc{cMSM}. 
We give a dynamic program filling a one-dimensional table $D_g$ in reverse order.
Without loss of generality, assume $m\ge n$. 
The idea is to align a point~$x_{m-i}$ to a point~$y_{n-i}$ for $i< n$. All remaining points $(x_1,\ldots, x_{m-n})$ are aligned to $y_1$. 
The first entry $D_g[1]$ corresponds to the distance $d_{greedy}(x,y)$.
The mechanism of the greedy heuristic is to merge two points $x_{m-(i+1)}$ and $x_{m-i}$, $i<n$, if their Euclidean distances to $y_{n-(i+1)}$ and $y_{n-i}$, respectively, are both greater than $2c$. 
Otherwise, a move from  $x_{m-(i+1)}$ to $y_{n-(i+1)}$ is applied. 
For $0<i<n$ we get the following recursion:
If  $\vert x_{m-(i+1)}-y_{n-(i+1)} \vert \geq2c$ and  $\vert x_{m-i}-y_{n-i} \vert \geq2c$, then
\begin{align*}
    D_g[m-(i+1)] = D_g[m-i] + 2c + \vert x_{m-(i+1)}-x_{m-i} \vert+\vert y_{n-(i+1)}-y_{n-i} \vert;
\end{align*}
otherwise, 
$D_g[m-(i+1)] = D_g[m-i] + \vert x_{m-(i+1)}-y_{n-(i+1)}\vert.$
For $n\leq i <m$ a move operation is beneficial if the distance between the regarded points is smaller than $c$. Furthermore, only merge and split operation are necessary to balance the lengths of the time series. 
Formally, if  $\vert x_{m-(i+1)}-y_1 \vert \geq c$ and  $\vert x_{m-i}-y_1 \vert \geq c$, then
\begin{align*}
    D_g[m-(i+1)] = D_g[m-i] + c  + \vert x_{m-(i+1)}-x_{m-i} \vert;
\end{align*}
otherwise, $D_g[m-(i+1)] = D_g[m-i] + c + \vert x_{m-(i+1)}-y_{1}\vert$. 
The first entry to be computed is $D_g[n] = \vert x_m - y_n \vert $. The computed value is an upper bound of the MSM distance as it  traverses one path of the DP table.   

\subsection{Sakoe-Chiba Band}
Computing every element of $D$ can be time-consuming, especially for large time series. 
A common global constraint is the Sakoe-Chiba band~\cite{sakoe1978dynamic}, which reduces the amount of entries that have to be computed. 
The idea is to narrow the space around the diagonal of the dynamic programming table. 
It is independent of the current row $i$. 
A parameter~$b$ is the absolute number of entries to be computed on the right and on the left side of a diagonal entry $D[i,i]$.
The overall bandwidth is $B=2b +1$, that is the maximum absolute coordinate deviation between two aligned points. For example, if $B=3$ then the alignment of $x_1$ to $y_4$ is allowed, the alignment of $x_1$ to $y_5$ is prohibited. Figure \ref{fig:Itakura_Sakoe} (left) shows a schematic example. 
For quadratic ($m\times m$)-tables, a band of size $B=1$ ($b=0$) corresponds to the Euclidean distance, a band of size $B=m$ to the exact distance. 

\begin{figure}
    \centering
    \includegraphics[width=0.35\textwidth]{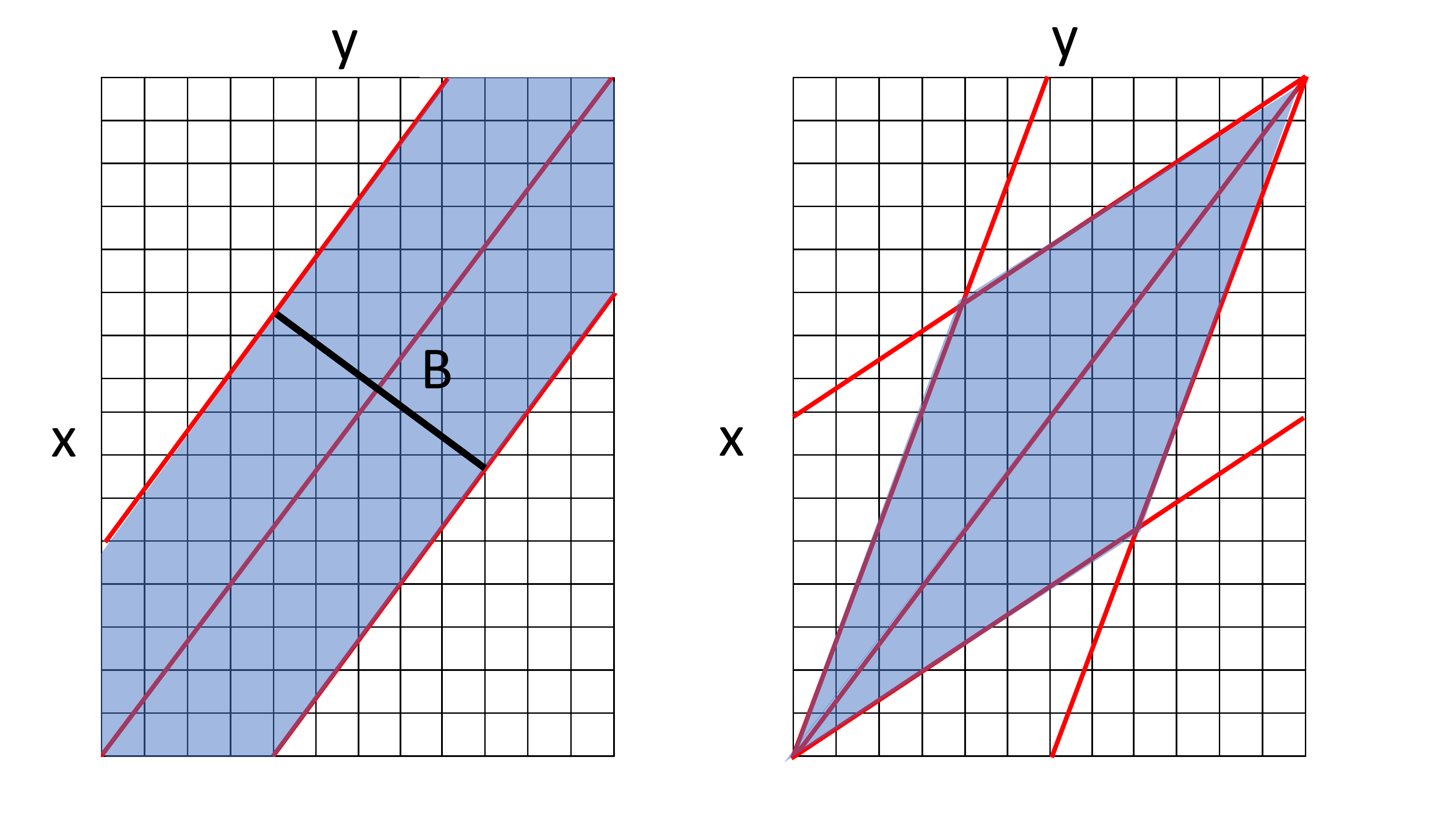}
    \caption{Schematic representation of the Sakoe-Band with $B=6$ (left) and an Itakura parallelogram with $d=\frac{1}{2}$ (right).}
    \label{fig:Itakura_Sakoe}
\end{figure}

\subsubsection{Slanted Band}
The Sakoe-Chiba band is appropriate  when $m\sim n$~\cite{giorgino2009computing}. If $\vert m-n \vert > B$, a band of size $B$ does not include the final coordinate $[m,n]$ in $D$. In this case, no solution exists. 
The Sakoe-Chiba band can be modified introducing the \emph{slanted band}. The new diagonal connects the entry $[0,0]$ and $[m,n]$ and has no longer a slope of 1 but of $m/n$. 

\subsection{Itakura Parallelogram}
A second common global constraint is the Itakura parallelogram~\cite{itakura1975minimum}. 
Again, the space around the diagonal of the dynamic table is narrowed but not with a band but with a parallelogram. 
That means, that in the beginning and the end the possible alignments are more restricted than in the middle of the table. 
Let $d\in(0,1]$ be the parameter that sizes the parallelogram. Setting $d=1$ corresponds to the Euclidean distance for a quadratic table. The smaller $d$ gets, the more entries are computed. Figure \ref{fig:Itakura_Sakoe} (right) illustrates the parallelogram resulting from $d=\frac{1}{2}$. 
The computation depends on the current row $i$. 
For each row~$i$ a new $start_i$ and $end_i$ coordinate has to be calculated. 
Given a ($m\times n$)-table, for $ i\in [m]$ we get 
\begin{align*}
    start_i &= \max\left(d \frac{n}{m} i, \frac{1}{d} \frac{n}{m} i - \frac{1-d}{d} n\right) \\
    end_i &= \min\left( \frac{1}{d} \frac{n}{m} i, d \frac{n}{m} i + (1-d)n \right).
\end{align*}

In the next section, we focus again on strategies to speed up the exact MSM computation. 
We consider pruning techniques to reduce the amount of entries to be computed.

\section{PrunedMSM}
\label{sec:pruningStrategies}
In this section we give the \emph{\textit{PrunedMSM}} algorithm for exactly computing the \textsc{MSM} distance between two time series. 
\textit{PrunedMSM} adapts the improved \textsc{MSM} by pruning table entries that do not lead to an optimal result. 
The algorithm follows a similar procedure as the $\textit{PrunedDTW}$ algorithm, proposed by Silva and Batista~\cite{SilvaB16}. 
A table entry $D[i,j]$ represents the optimal cost transforming $(x_1,\ldots x_i)$ to $(y_1,\ldots,y_j)$, that includes an alignment of point~$x_i$ and $y_j$. If this entry exceeds a certain value, it is likely that in the final distance computation of $x$ and $y$, the points~$x_i$ and $y_j$ will not be aligned since otherwise the transformation cost are not optimal. 
More precisely, an \emph{upper bound} ($UB$) of the exact distance is computed in advance. 
If the value of an entry is greater than the upper bound, the alignment can not be part of the optimal transformation. 
To compute the value of an entry~$D[i,j]$, the entries~$D[i-1,j-1]$, $D[i-1,j]$ and $D[i,j-1]$ are considered. If all these three entries have a value greater than the upper bound, the value of the entry~$D[i,j]$ also has to be greater than the upper bound. 

We now want to know which entries in a row~$i$ do not need to be considered.
To separate these entries from the relevant ones, a start parameter~$sc$ and an end parameter~$ec$ are updated for each row. 
We traverse a row~$i$ from left to right. 
Figure \ref{fig:PruningStrategy} (left) shows the update strategy for parameter~$sc$.
As long as all entries that are computed are greater than the upper bound, all entries in the next row with the same index $j$ are also greater than the upper bound. In the example in row ~2, $sc$ is set to 1. All other entry below $sc=1$ does not have to be computed because they are guaranteed to be greater than the upper bound.

The second parameter~$ec$ defines where to abort the computation of the table entries in the next row~$i+1$. 
Assume $ec = j$ is set in row $i$, if $D[i+1,j]>UB$, the following entry $D[i+1,j+1]$ does not have to be computed, because all entries to be considered are greater than the upper bound. The parameter~$ec$ is updated for the next row in that way, that it equals the index of the first entry in the row that is greater than the upper bound such that all following entries until position $j$ are also greater than the upper bound. 
Figure~\ref{fig:PruningStrategy} (right) shows an example of the update strategy of the parameter $ec$. 
After computing row 2, $ec$ is set to 3. For row~3, the entries $[3,2]$ and $[3,3]$ are greater than the upper bound. Since $ec=3$, the entry $[3,4]$ will not be computed. After computing row~3, $ec$ is thus set to 2.  

\begin{figure}
    \centering
    \includegraphics[width=0.5\textwidth]{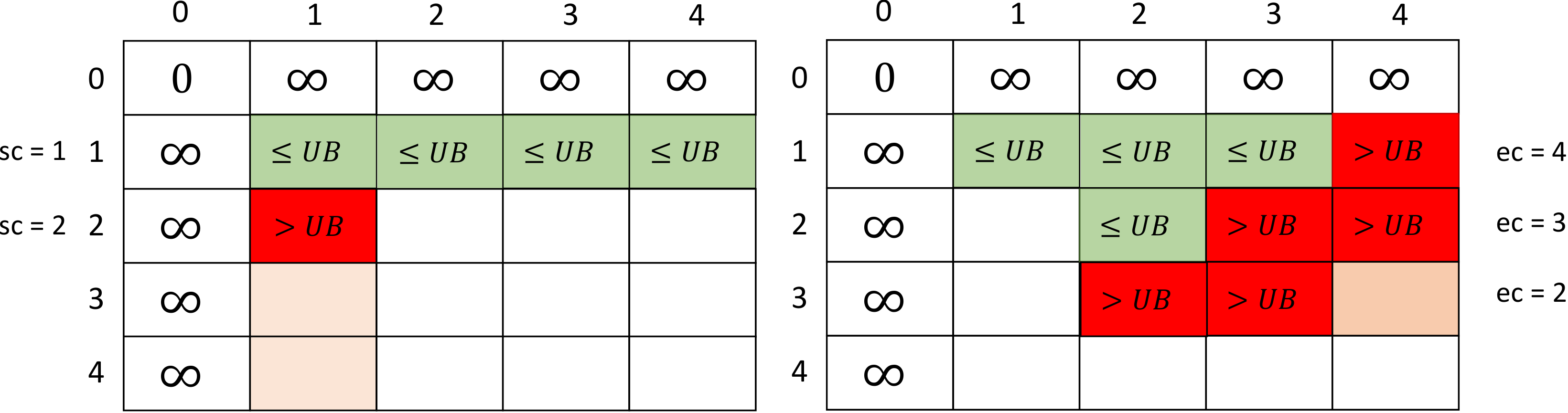}
    \caption{Representation of the pruning strategy of \textsc{MSMPruned}. Left: Setting the parameter $sc$, right: setting the parameter $ec$.}
    \label{fig:PruningStrategy}
\end{figure}

We now develop several bounding strategies.

\subsection{Upper Bounds}
To achieve good pruning results, it is crucial to determine appropriate upper bounds. 
The smaller the upper bound is, the more likely the entries of the dynamic programming table are pruned. 
We can use all presented heuristics in Section~\ref{sec:Heuristics} to receive different upper bounds. 
 
\subsubsection{Updating the Upper Bound}
\label{updatingUB}
The dynamic tables of the Sakoe-Chiba band heuristic and the Itakura parallelogram heuristic store the optimal value aligning two points on the (slanted) diagonal. 
For the triangle heuristic a similar table can be built. 
The greedy heuristics stores the values for a shifted diagonal with slope~1 and the following vertical or horizontal alignment.
Calculating these (diagonal) values in reverse order, makes it possible to update the upper bound every time these entries are computed.  
Let $R$ be a one-dimensional table of size $m$ storing these values of a heuristic. 
The upper bound~$UB$ for a quadratic table can be updated after setting the value for $D[i,i]$:
$UB = D[i,i] + R[i+1].$

All values that will be guaranteed to be greater than the upper bound will not be computed. 
To reduce the amount of entries to be computed even more, we introduce two lower bound strategies in the next section. 

\subsection{Lower Bounds}
Computing an entry~$D[i,j]$ gives the optimal cost to transform $(x_1,\ldots x_i)$ into $(y_1,\ldots,y_j)$.
We will now estimate the cost for an entry~$D[i,j]$ for the remaining alignment of $(x_{i+1},\ldots x_m)$ to $(y_{j+1},\ldots,y_n)$.
These remaining cost are a \emph{lower bound} ($LB[i,j]$) of the optimal cost. 
For each entry, we get an estimation $E$ for the minimum total cost: $ E[i,j] = D[i,j] +  LB[i,j]$.
We now check for each entry, if $E[i,j] > UB$.  
Hence, introducing a lower bound increases the likelihood of pruning more table entries.
The most important property of the estimation $E[i,j]$ is that it never overestimates the optimal transformation cost.  
In the following, we introduce two lower bound strategies. 

\subsubsection{LB$_{ms}$}
The first lower bound $LB_{ms}$ counts the amount of remaining merge or split operations. 
We have $LB_{ms}[i,j] = \vert m-i - (n-j) \vert \cdot c$.
For example, let $x = (x_1, x_2,x_3)$ and $y=(y_1, y_2, y_3, y_4)$. For the position $[1,4]$, at least two merge operations in $x$ are needed to align $x_2$ and $x_3$ to $y_4$, that is, at least cost of~$2c$ are added to align $x$ to $y$.

\subsubsection{LB$_{t}$}
\label{sec:LBtriangle}
The second lower bound $LB_{t}$ makes use of the \textsc{cMSM} algorithm and the lower triangle inequality.
The computation is similar to the computation of the triangle heuristic in Section \ref{sec:TriangleHeuristic}.
Assume a constant time series $q$, the \textsc{MSM} distances~$d(x,q^{(m)})$ and $d(y,q^{(n)})$ between~$q$ and~$x$, and~$q$ and~$y$, respectively, is computed in reverse order. 
Applying the lower triangle inequality results in the following estimation: $d(x,y) \geq \vert d(x,q^{(m)}) - d(y,q^{(n)}) - (m-n)\cdot c\vert$.
Further the intermediate cost aligning $(x_i,\ldots,x_m)$ to $(q_i,\ldots,q_m)$ are stored in $D_c^x[i]$ of an one-dimensional table~$D_c^x$ of size~$m$.
The entry $D_c^x[1]$ corresponds to the distance~$d(x,q^{(m)})$. 
Analogously, all intermediate cost for $y$ and~$q^{(n)}$ are stored in~$D_c^y$. 
We now apply the lower triangle inequality for every entry~$[i,j]$. 
The minimum remaining cost of an entry $[i,j]$ considers the  alignment of $(x_{i+1} ,\ldots, x_m)$ to $(y_{j+1},\ldots,y_n)$. 
We get $d((x_{i+1} ,\ldots, x_m),(y_{j+1},\ldots,y_n)) \geq \vert D_c^x[i+1] - D_c^y[j+1] - (m-(i+1)-(n-(j+1)))\cdot c \vert $. 
Adding the right part of this inequality to~$D[i,j]$ to get a lower bound would overestimate the optimal transformation cost.
At the transition from $(x_i, x_{i+1})$ to $(y_j,y_{j+1})$, these points may be merged which saves move cost while creating additional merge costs. 
The table~$S$ includes potential move cost at this transition, that has to be subtracted, that is $a = \vert x_i-q\vert + \vert y_j-q\vert$, and respective merge/split cost that has to be added: 
\begin{equation*}
S[i,j]= \begin{cases}
\max(a-c,0) & 1<i<m, 1<j<n \\
0, &\text{if } i=m, j=n\\
a, & \text{otherwise.}
\end{cases}
\end{equation*}
Now, the lower bound~$LB_t$ is computed as follows:
\begin{align*}
    LB_t[i,j] &= \vert D_c^x[i+1] - D_c^y[j+1] \\
    &- (m-(i+1)-(n-(j+1))) \cdot c \vert -S[i,j].
\end{align*}

\subsection{Further Improvement}
\label{pruningBand}
To prune even more table entries, introduce another \emph{pruning band}. 
The pruning band is similar to the idea of the Sakoe-Chiba band and the Itakura parallelogram. For each row, the maximum absolute coordinate deviation between two aligned points is computed.
Moving horizontally in the dynamic programming table means that a split operation is applied. 
An upper bound gives a limitation for the maximum amount of split operations that can be applied. 
More precisely, setting a bandwidth $b=\lceil UB/c \rceil$ prunes the table while still giving the optimal cost transformation. 

In the next section, we will evaluate \textit{PrunedMSM}, all presented heuristics and benchmark the fastest exact \textsc{MSM} algorithm with the \textit{PrunedDTW}. 

\section{Experiments}
\label{sec:experiments}

The running times of our Java implementations\footnote{All code is available on GitHub: \url{https://github.com/JanaHolznigenkemper/msm_speedup.git}} are measured on a computer with Ubuntu Linux 20.04, AMD Ryzen 7 2800 CPUs, 32GB of RAM, Java version 15.0.1.  

\subsection{Data}
We performed our experiments on 117 data sets of the UCR archive~\cite{UCRArchive2018} containing time series of equal lengths.  
All tested algorithms also work for time series of unequal lengths. The parameter $c$ for merge and split cost is set to $0.5$ as in the study of Paparrizos et al.~\cite{PaparrizosLEF20}.
\subsection{Heuristics} 
We first evaluate the accuracy and running time of the proposed heuristics. 
All running times are compared to the improved \textsc{MSM} algorithm.
As a first result, the improved \textsc{MSM} algorithm is 3.2\% faster on average over all data sets than the original implementation. 
For the triangle heuristic it is crucial to find a constant time series $q^{(m)}$ such that the distance to both time series is as small as possible. 
Since all time series are normalized with a mean of 0, we choose $q=0$. 
We test different band sizes ($b = 0\%, 10\%,20\%)$ for the Sakoe-Chiba heuristic relative to the time series input length $m$.
The parameter $d$ for the Itakura heuristic is set to $d\in(\frac{1}{2}, \frac{2}{3}, \frac{3}{4})$.
The relative error is the relative deviation of the distance computed by a heuristic and the original distance. 

\begin{figure}[t]
    \centering
    \includegraphics[width=0.43\textwidth]{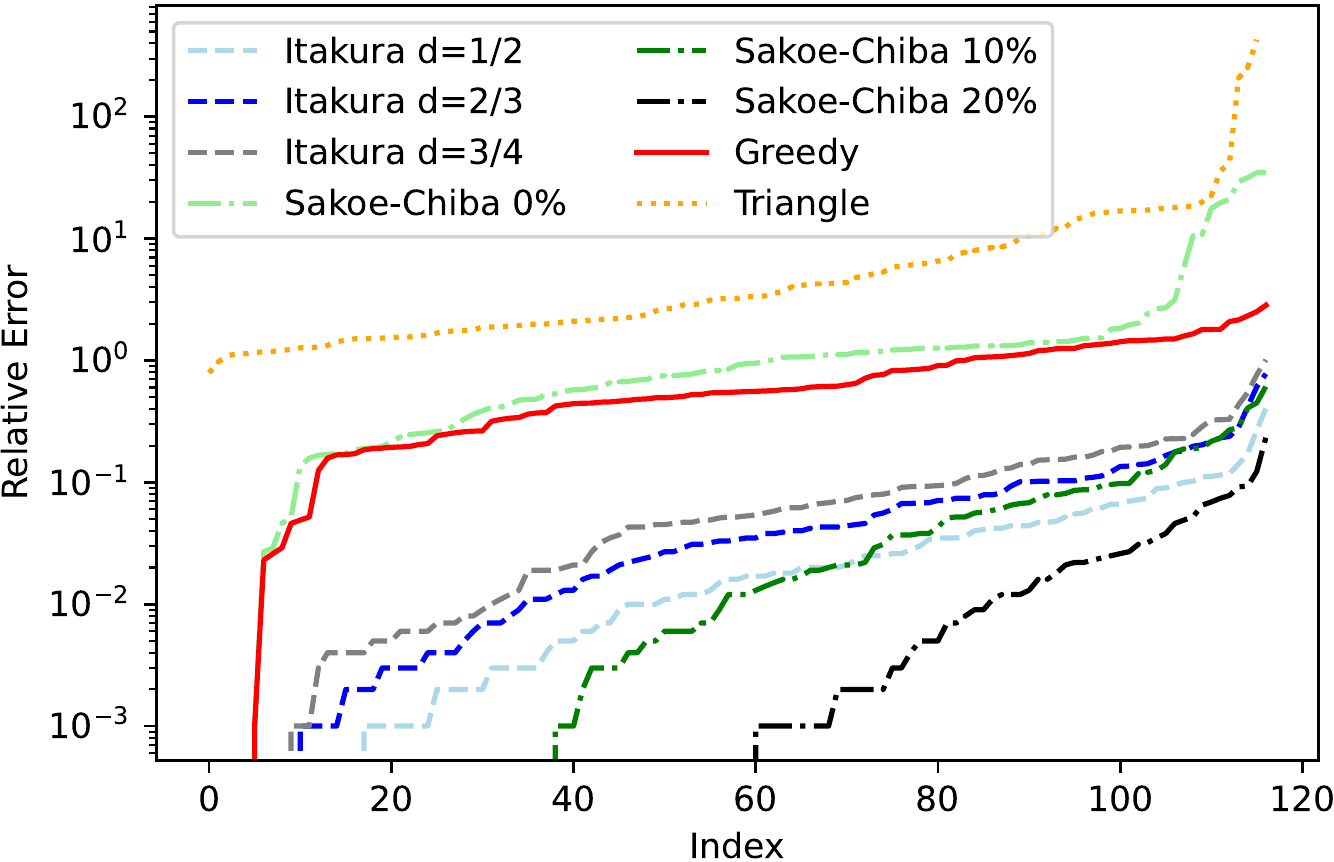}
    \caption{Average relative error per data set computing the \textsc{MSM} by the proposed heuristics. The results are sorted by the relative error for each heuristic.} 
    \label{fig:allHeuristics}
\end{figure}
Figure \ref{fig:allHeuristics} gives an overview of the tested heuristics. 
For each algorithm, the relative error is averaged by data set and sorted in an ascending manner, that is, index $i$ corresponds to the data set with the $i$-th best relative error. 
The most accurate results are achieved by the Sakoe-Chiba heuristic for $b= 10\%,20\%$ and by the Itakura heuristic for $d=\frac{1}{2}$. 
The opposite applies for the running time. 
Figure \ref{fig:allHeuristicsRT} gives the average running times of the heuristics per data set, again sorted by running time for each heuristic.
The approaches with the best accuracy results are much slower than other heuristics.

\begin{figure}[t]
    \centering
    \includegraphics[width=0.43\textwidth]{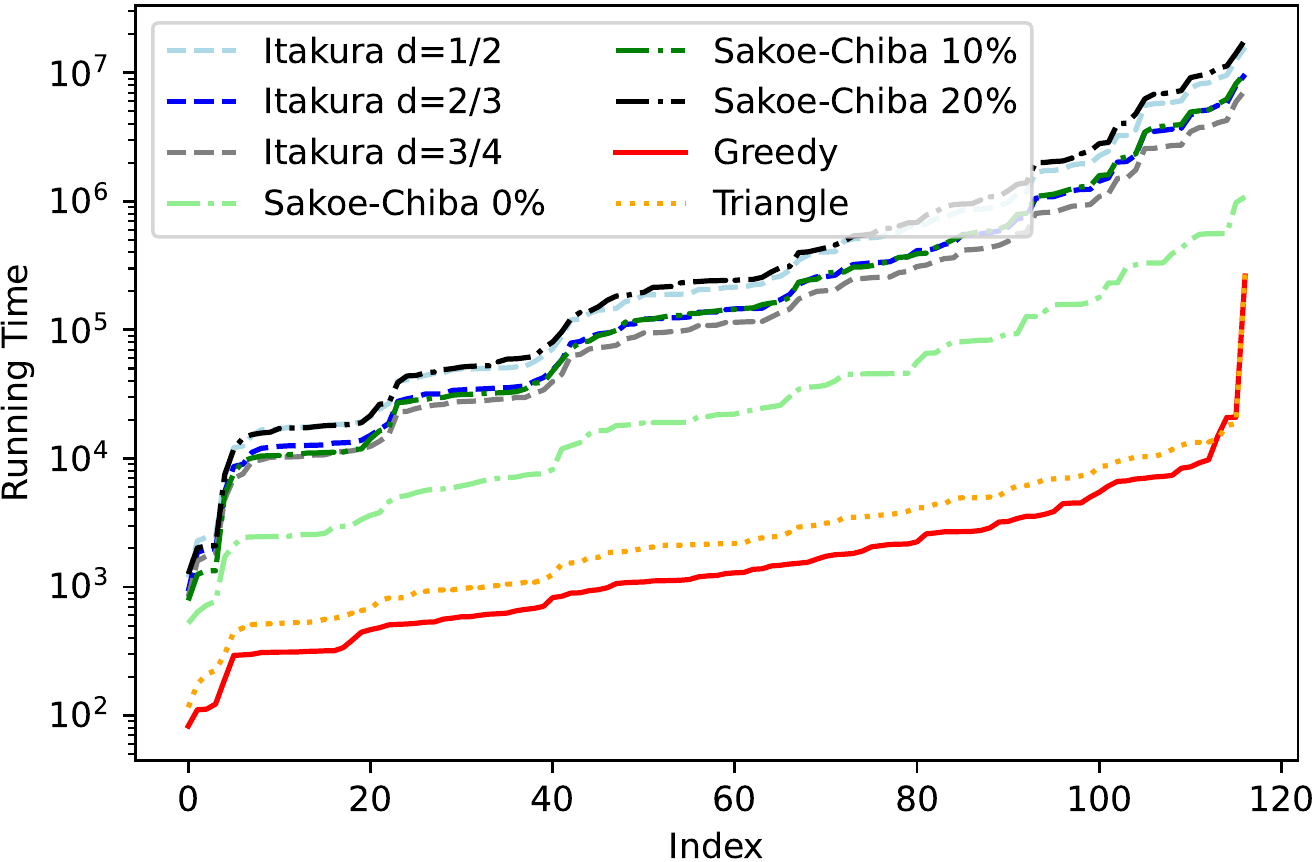}
    \caption{Average running time in ns per data set computing the \textsc{MSM} by the proposed heuristics. The results are sorted by running time for each heuristic.}
    \label{fig:allHeuristicsRT}
\end{figure}

Regarding the trade-off between running time and accuracy, the Sakoe-Chiba heuristic with a band size of $b=10\%$ or the Itakura parallelogram with $d=\frac{2}{3}, \frac{3}{4}$ seem the best choice when focusing on accuracy; the greedy heuristic seems the best choice when focusing on running time.

 \begin{figure}[t]
    \centering
    \includegraphics[width = 0.45\textwidth]{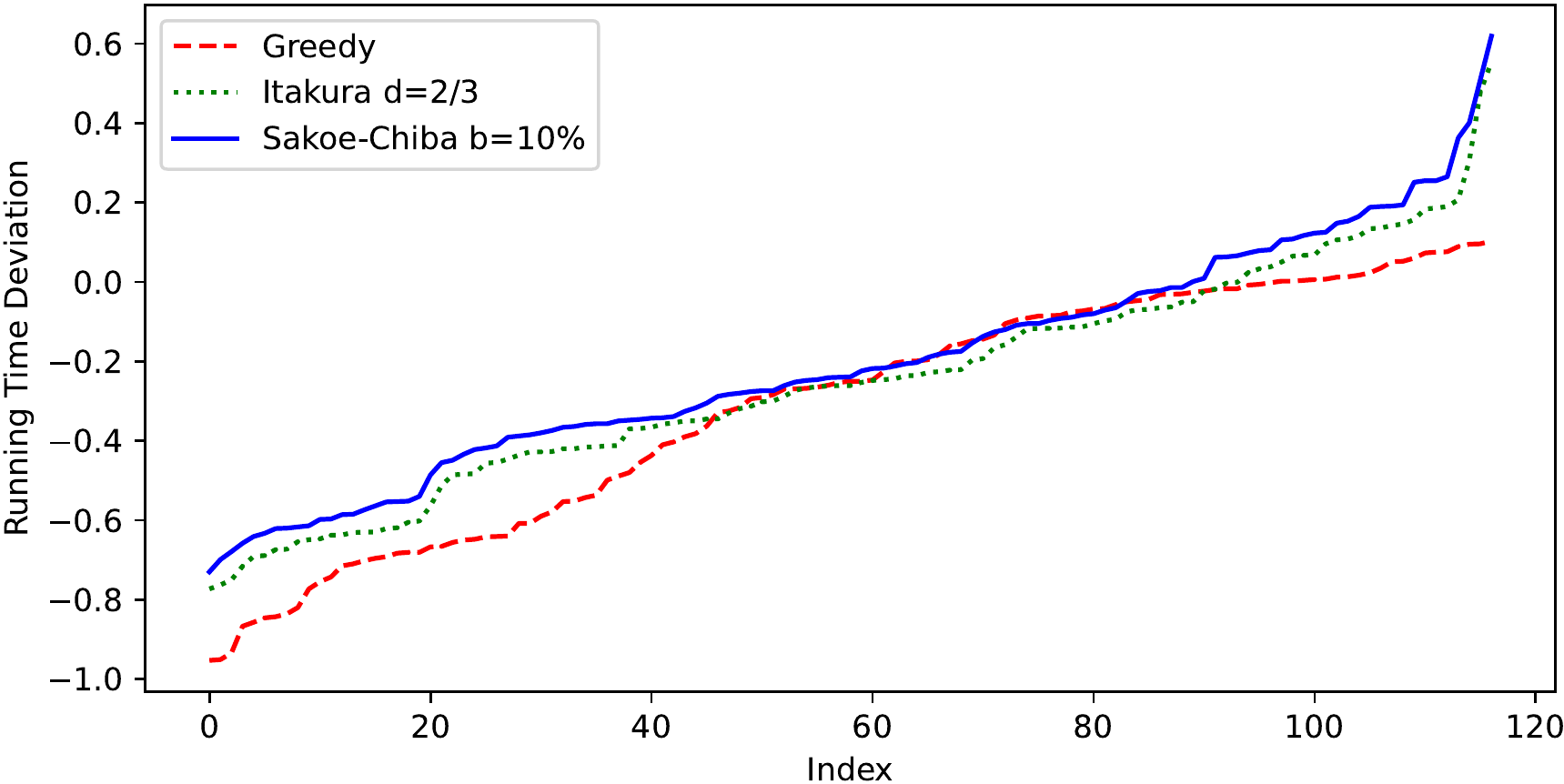}
    \caption{Average running time deviation per data set for the \textit{PrunedMSM} algorithm using the lower bound $LB_t$ and different upper bounds. The results are sorted by the deviation for each heuristic.}
    \label{fig:SimpleLB}
\end{figure}
 \subsection{Pruning}
 In the following, we test \textit{PrunedMSM} against the improved \textsc{MSM} of Section \ref{sec:Speedup}. We test three different upper bounds for the lower bound $LB_{ms}$
 The first is  the greedy heuristic since it is the fastest heuristic. We further take a Sakoe-Chiba band with $b=10\%$ of the time series lengths and an Itakura parallelogram with $d=\frac{2}{3}$. Both approaches are slower than the greedy heuristic but achieve good results regarding accuracy.

 Figure \ref{fig:SimpleLB} shows the running time deviation compared to the improved \textsc{MSM} implementation. The results are sorted by deviation. \textit{PrunedMSM} with an upper bound given by the greedy heuristic has the best running time: 
the greedy upper bound performs best for 78 data sets, the Itakura ($d=\frac{2}{3}$) upper bound for 26 data sets, and the Sakoe-Chiba band ($b=10\%$) for 13 data sets. 
 We will not consider the Sakoe-Chiba band for our further experiments.
 
 We next evaluate the influence of updating the upper bound as described in Section~\ref{updatingUB} and inserting the pruning band of Section \ref{pruningBand}. 
Figure \ref{fig:simpleLBUB} shows the running time deviation compared to the improved \textsc{MSM} for  Itakura ($d=\frac{2}{3}$) and the greedy heuristic for updated bounds with pruning band and non-updated upper bounds.
The running time results without pruning band 
are similar to the ones with band but slightly worse. 
The graph shows a clear running time advantage for the updated greedy upper bound.
Table \ref{tab:pruning} summarizes the number of data set for which a certain pruning strategy performs best. 

\begin{figure}[t]
    \centering
    \includegraphics[width=0.45\textwidth]{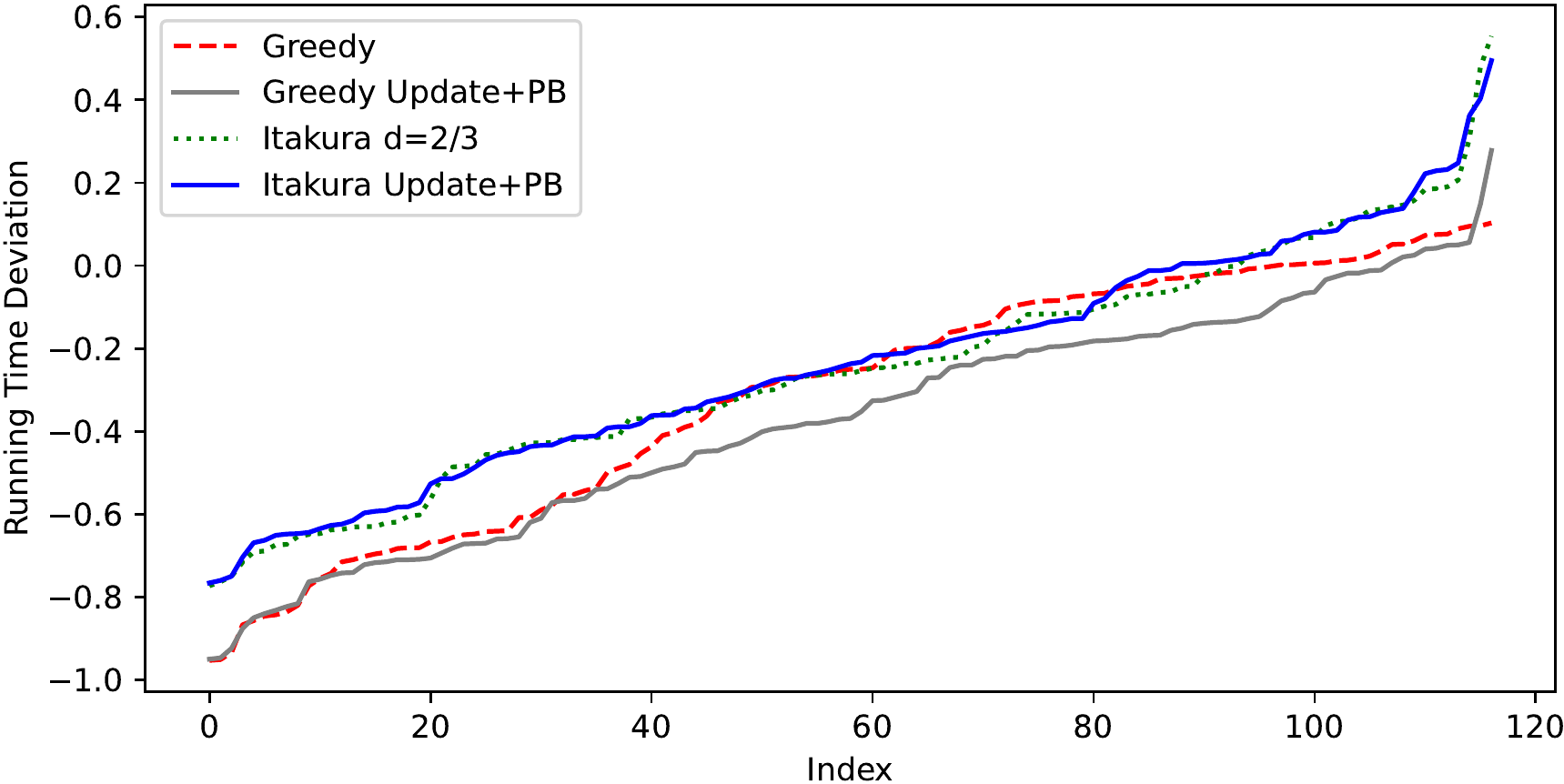}
    \caption{Average running time deviation per data set for \textit{PrunedMSM} using the lower bound $LB_t$ and the greedy or Itakura upper bound with and without updating the upper bound and inserting a pruning band (PB).  The results are sorted by deviation for each heuristic.}
    \label{fig:simpleLBUB}
\end{figure}

\begin{table}
\caption{Number of data sets for which the given pruning strategy performs best (g=greedy, Itak=Itakura $d=\frac{2}{3}$, U=Update Upper Bound, B=Pruning Band).}
\label{tab:pruning}
\centering
\resizebox{0.4\textwidth}{!}{%
  \begin{tabular}{|c|c|c|c|c|c|} 
 \hline
g & g+U & g+U+B &  Itak & Itak+U & Itak+U+B \\
 \hline
 28 & 19 & 59 & 5 & 1 & 5 \\ 
 \hline
\end{tabular}}
\end{table}

We further test the performance of \textit{PrunedMSM} selecting always the maximum lower bound of  $LB_t$ and $LB_{ms}$. First, we compare $LB_t$ to $L_{ms}$.
For 34.6\% of all computed lower bound entries $[i,j]$ $LB_t[i,j]>LB_{ms}[i,j]$. Compared to the running time of  \textit{PrunedMSM} with only $LB_{ms}$ there is no improvement since the additional number of operation per loop is not compensated by the further space reduction. 

\subsection{DTW Comparison}
Finally, we compare the fastest \textit{PrunedMSM} algorithm, that is, \textit{PrunedMSM} with greedy upper bound including updates and the pruning band and the $LB_{ms}$, to a state-of-the-art DTW distance computation, \textit{PrunedDTW}.
Since \textit{PrunedDTW} is implemented in C++, we implemented \textit{PrunedMSM} in C++ for a meaningful comparison.  
The C++ code was compiled with the GNU C++ compiler (g++) using the -O3 optimization flag.
Figure \ref{fig:DTWvsMSM} shows the average running times per data sets. For 94 out of 117 data sets, \textit{PrunedMSM} achieves better running times than \textit{PrunedDTW}.  

\begin{figure}[t]
    \centering
    \includegraphics[width=0.4\textwidth]{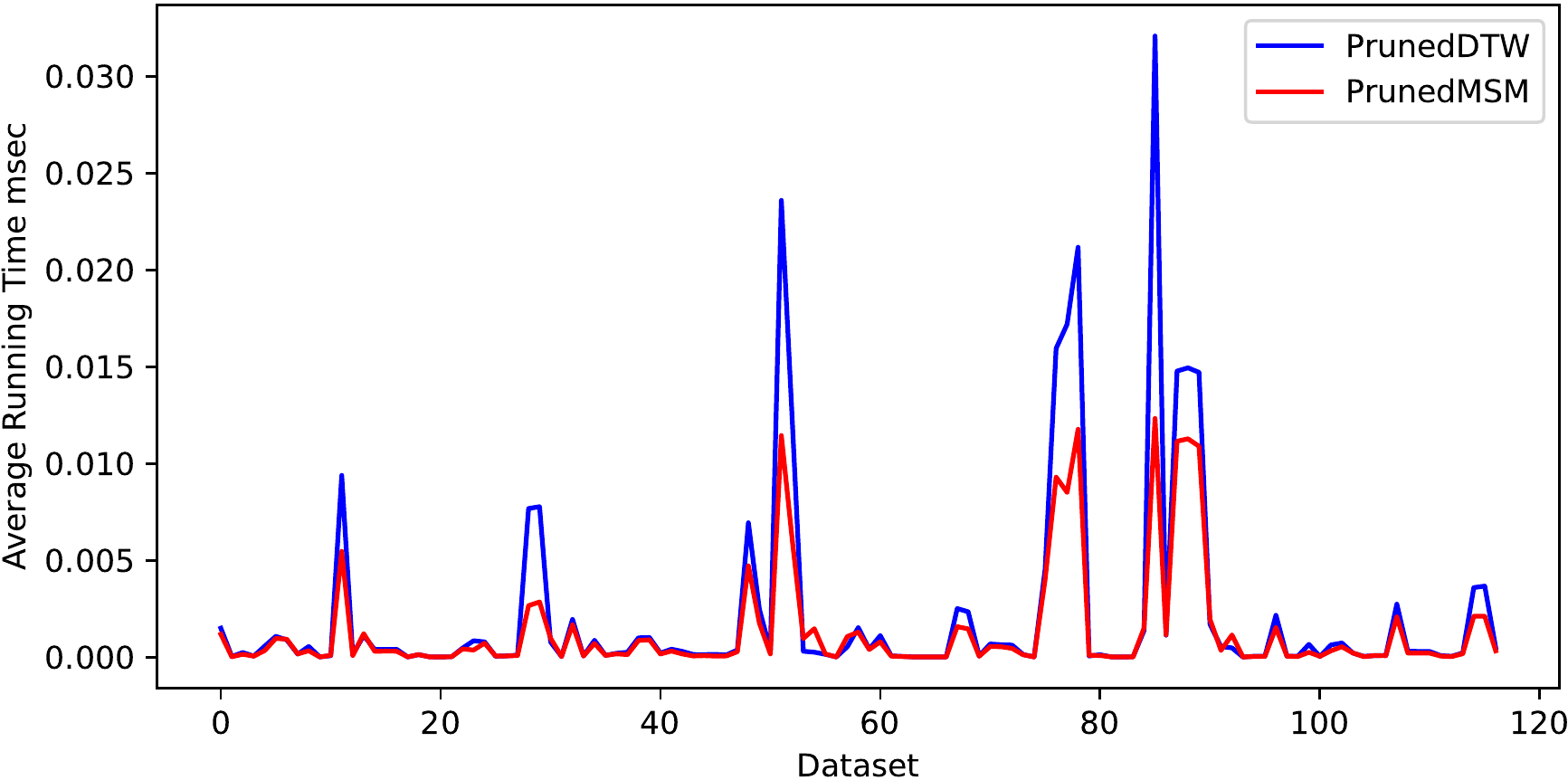}
    \caption{Average running time per data set for the \textit{PrunedMSM} and \textit{PrunedDTW} algorithms.}
    \label{fig:DTWvsMSM}
\end{figure}

\section{Conclusion and Future Work}
\label{sec:Conclusion}
This paper introduces several heuristics and pruning strategies to speed up the computation of the Move-Split-Merge (\textsc{MSM}) metric. 
Experimental results show good accuracy and excellent running time advantages of the proposed heuristics. 
Moreover, we achieved to speed up the computation of exact \textsc{MSM} by introducing \textit{prunedMSM} so that it is now faster than a state-of-the-art \textsc{DTW} distance computation for a majority of the popular UCR data sets. 

In future work, we will first investigate the accuracy of the proposed heuristics regarding classification tasks, like 1-NN classification. 
This may include an analysis of the impact of the parameter~$c$ regarding running time and accuracy. 
Second, we plan to extend the \textsc{cMSM} algorithm to obtain linear-time algorithms for computing the distance between arbitrary time series and structured time series, for example piecewise-linear time series. 
Third, we will analyze the use of \textsc{PrunedMSM} for similarity search where one fixed time series is compared to a large set of time series arriving in a stream, similar to the approach of the UCR suite~\cite{SilvaGKB18}. 

\bibliographystyle{abbrv}
\bibliography{references}  






\appendix

\section{Properties of Transformation Graphs}
\label{app:propertiesTrafoGraphs}
In the following, we summarize some important known properties about the transformation graph by Stefan et al. \cite{StefanAD13} and Holznigenkemper et al. \cite{holznigenkemper2023computing}.  
The first lemma states that there is always an optimal monotonic transformation.

\begin{lemma}[Monotonicity lemma \cite{StefanAD13}]
\label{lemma:monotonicity}
For any two time series $x$ and $y$, there exists an optimal transformation that converts $x$ into $y$ and that is monotonic. 
\end{lemma}

An important result states that there always exists an optimal transformation graph only containing paths from source to sink nodes of the following consecutive edge types \cite{holznigenkemper2023computing}: 
\begin{description}
    \item[Type 1:] move - move - $\cdots$ - move - move 
    \item[Type 2:] split/move  - $\cdots$ -split/move
    \item[Type 3:] merge/move -  $\cdots$ - merge/move
    \item[Type 4:] Type 3 - merge - move - split - Type 2
\end{description}

It follows that there exists a transformation graph which can be decomposed in its weakly connected component that fulfill the properties of a tree. 
We summarize the possible structures of these trees:

\begin{description}
    \item[Type-1-Tree:] contains only paths of Type 1, that is, there is only one move edge in the tree connecting one source and sink node (see Figure \ref{fig:tree_defs}a). 
    \item[Type-2-Tree:] contains only paths of Type 2. It has only one source node and at least two sink nodes (see Figure \ref{fig:tree_defs}b).
    \item[Type-3-Tree:] contains only paths of Type 3. It has at least two source nodes whose paths reach the same sink node (see Figure \ref{fig:tree_defs}c).
    \item[Type-4-Tree:] contains only paths of Type 4. It has at least two source and two sink nodes (see Figure \ref{fig:tree_defs}d). 
\end{description}

\begin{figure}[b]
    \centering
    \includegraphics[width=0.5\textwidth]{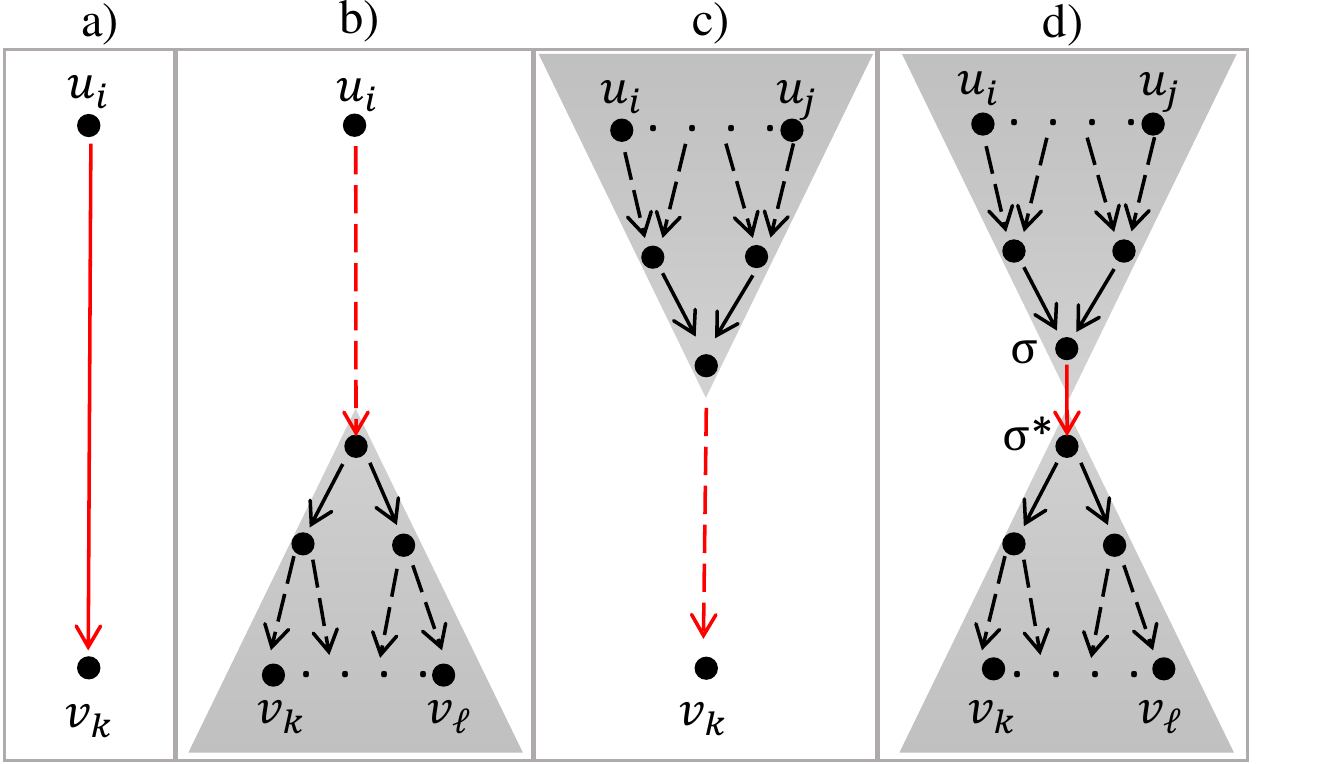}
    \caption{All red edges are move operations. Black arrows are merge or split edges. The dashed lines represent paths from one node to another without specifying how many intermediate node are on them.  a) Type-1-Tree b) Type-2-Tree c) Type-3-Tree d) Type-4-Tree.}
    \label{fig:tree_defs}
\end{figure}

The following lemma states that there always exists an optimal transformation graph, where every weakly connected component is a tree of Type~1--4. 

\begin{lemma}{\cite{holznigenkemper2023computing}}
Let $x$ and $y$ be two time series. Then there exists an optimal transformation graph $G_{\mathds{S}}(x,y)$ such that its weakly connected components are only trees of Type~1--4. 
\end{lemma}
It follows that we can decompose an optimal transformation graph $G_{\mathds{S}}(x,y)$ into a sequence of distinct trees $(\mathcal{T}_1,\ldots,\mathcal{T}_t)$.
Each tree $\mathcal{T}_i$ has a set of sink nodes $N_{\mathcal{T}_i}(x)$ and a set of source nodes $N_{\mathcal{T}_i}(y)$.
All nodes of $N_{\mathcal{T}_i}(x)$  and $N_{\mathcal{T}_i}(y)$ are successors of $N_{\mathcal{T}_{i-1}}(x)$ and $N_{\mathcal{T}_{i-1}}(y)$, respectively.
We call a tree \emph{monotonic} if all paths in the tree are monotonic. 
Further a tree may be specified as \emph{increasing} or \emph{decreasing}. 
Two trees are \emph{equivalent} if they have the same set of source and sink nodes. 
The \emph{cost of a tree $\mathcal{T}$} is the sum of the cost of all edges in the tree.\\

In the following, we denote an optimal transformation graph fulfilling these decomposition properties as an \emph{optimal transformation forest}.

\section{\textsc{cMSM}}
\label{app:cMSM}
The exact algorithm \textsc{cMSM} computes the distance between a constant time series, i.e., a time series which data points are equal to a constant $q\in \mathds{R}$, and an arbitrary time series $x$ of same length. 
Let $q^{(n)}=(q_1,\ldots, q_n)$ be a constant time series of length $n$  such that it holds for a constant $q$ that $q_i = q$ $\forall i\in [n]$. 
To compute the distance $d(x,q)$, we decompose the transformation forest $G_{\mathds{S}}(x,q^{(n)})$ into a sequence of distinct trees $(\mathcal{T}_1,\ldots,\mathcal{T}_s)$. 
To this end, we first determine an upper threshold $t^+$ such that $t^+>q$ and a lower threshold $t^-<q$. 
The idea is, that all points of $x$ that are above or below the upper or the lower threshold, respectively, are mostly contained in Type-4-Trees. 
All other points are in Type-1-Trees. 
The next lemma formalizes this decomposition. 

\begin{lemma}
\label{lemma:constantTsDecomposition}
Given a time series $x=(x_1,\ldots,x_n)$, two constants $c$ and $q$, and a constant time series $q^{(n)}$, $q_i=q$  $\forall i\in[n]$.
Then there exists an optimal transformation where each tree $\mathcal{T}$ is
\begin{enumerate}
\item a Type-4-Tree and $\forall u_j \in \mathcal{N}_{\mathcal{T}_i}(x): x_j\leq t^-$ or $\forall u_j \in \mathcal{N}_{\mathcal{T}_i}(x): x_j\geq t^+$ and $\vert \mathcal{N}_{\mathcal{T}_i}(x)\vert>1$,
\item or a Type-1-Tree
\end{enumerate}
for $t^+ = q+2c$ and $t^-=q-2c$. 
\end{lemma}

\begin{proof}
Let $G_{\mathds{S}}(x,q^{(n)})$ be an optimal transformation forest. 
Before starting the proof, we make some general observations regarding the alignment from $x$ to $q^{(n)}$. 
By Lemma \ref{lemma:monotonicity}, a merge of two consecutive points $x_i$ and $x_{i+1}$ is not optimal if $x_i<q$ and $x_{i+1}>q$ or $x_i>q$ and $x_{i+1}<q$. 
Since the time series are of equal length, the number of merge operation is equal to the number of split operations in $G$
 We assume that the number of source nodes is equal to the number of sink nodes in a tree $\mathcal{T}$. If this is not the case, split operations can be easily shifted to another tree without changing the cost. 
 
The cost of a Type-4-Tree $\mathcal{T}$ are at least the sum of the maximum distance from $q^{(n)}$ to the source node set of $\mathcal{T}$ and the respective merge and split cost. Formally, we get $\cost(\mathcal{T})\leq \vert max_{x_i\in N_{\mathcal{T}}(x)}(x_i) - q\vert + \vert N_{\mathcal{T}}(x)\vert c + \vert N_{\mathcal{T}}(q)\vert c $.\\
In the first part of the proof, we show that it is not optimal to merge two consecutive points where one point is above a threshold and the other one is below.

Without loss of generality we regard the upper threshold $t^+$ and a decreasing tree $\mathcal{T}$ with source node set $ N_{\mathcal{T}}(x) = (u_i,\ldots, u_j, u_{j+1}, \ldots u_k)$. Assume towards a contradiction that $\forall x \in (x_i,\ldots,x_j): x> t^+$ and $x_{j+1}<t^+$.
The nodes $(u_i,\ldots, u_j)$ merge to the intermediate node $\delta$. By Lemma XY citePaper, it holds for decreasing trees that $\delta = \min(x_i,\ldots,x_j) > x_{j+1}$. 
Let $\mathcal{T}_1$ and $\mathcal{T}_2$ be two decreasing trees such that $ N_{\mathcal{T}_1}(x) = (u_i,\ldots, u_j)$ and $ N_{\mathcal{T}_2}(x) = (u_{j+1}, \ldots u_k)$. 
Splitting the tree $\mathcal{T}$ into $\mathcal{T}_1$ and $\mathcal{T}_2$ changes the cost of $G$. The move cost from $\delta$ to $x_{j+1}$ and the merge and split cost for $x_{j+1}$ are subtracted. Additionally, the move cost from $\delta$ to $q$ are added. In total the cost of $\mathcal{T}_1$ and $\mathcal{T}_2$ are $\cost(\mathcal{T}_1) +\cost(\mathcal{T}_2)= \cost(\mathcal{T}) - (\delta - x_{j+1}) - 2c + (\delta - q)$. 
It follows that $\cost(\mathcal{T}_1) +\cost(\mathcal{T}_2) < \cost(\mathcal{T})$ if $x{j+1} > q +2c$. 
Setting $t^+ = 2c + q$ we get a contradiction to our assumption that $G$ is optimal.

Second, we prove that a sequence consecutive of points larger than a threshold are in one tree. 
Without loss of generality, let $x' = (x_i,\ldots,x_j,x_{j+1},\ldots,x_k)$ be a sequence of consecutive points such that for all $x~\in~x': x\geq t^+$. 
Assume towards a contradiction that $\mathcal{T}_1$ and $\mathcal{T}_2$ are two decreasing trees with source nodes  $N_{\mathcal{T}_1}(x) = (u_i,\ldots,u_j)$ and $N_{\mathcal{T}_2}(x) = (u_{j+1},\ldots,u_k)$. 
The cost for both trees are $\cost(\mathcal{T}_1) + \cost(\mathcal{T}_2) \geq (k-i-2)2c + \max(N_{\mathcal{T}_1}(x)) + \max(N_{\mathcal{T}_2}(x)) - 2q$.
Let $\mathcal{T}$ be a decreasing tree such that $N_{\mathcal{T}}(x) = (x_i,\ldots,x_k)$ with $\cost(\mathcal{T}) = (k-i-1)2c + \max(N_{\mathcal{T}_1}(x))-q$. 
If $2c< min(\max(N_{\mathcal{T}_1}(x)), \max(N_{\mathcal{T}_2}(x)))$ the decomposition into $\mathcal{T}_1$ and $\mathcal{T}_2$ is not optimal. Setting $t^+ = 2c + q$ we get a contradiction to this assumption. 
\end{proof}

\begin{figure}
    \centering
    \includegraphics[width=0.45\textwidth]{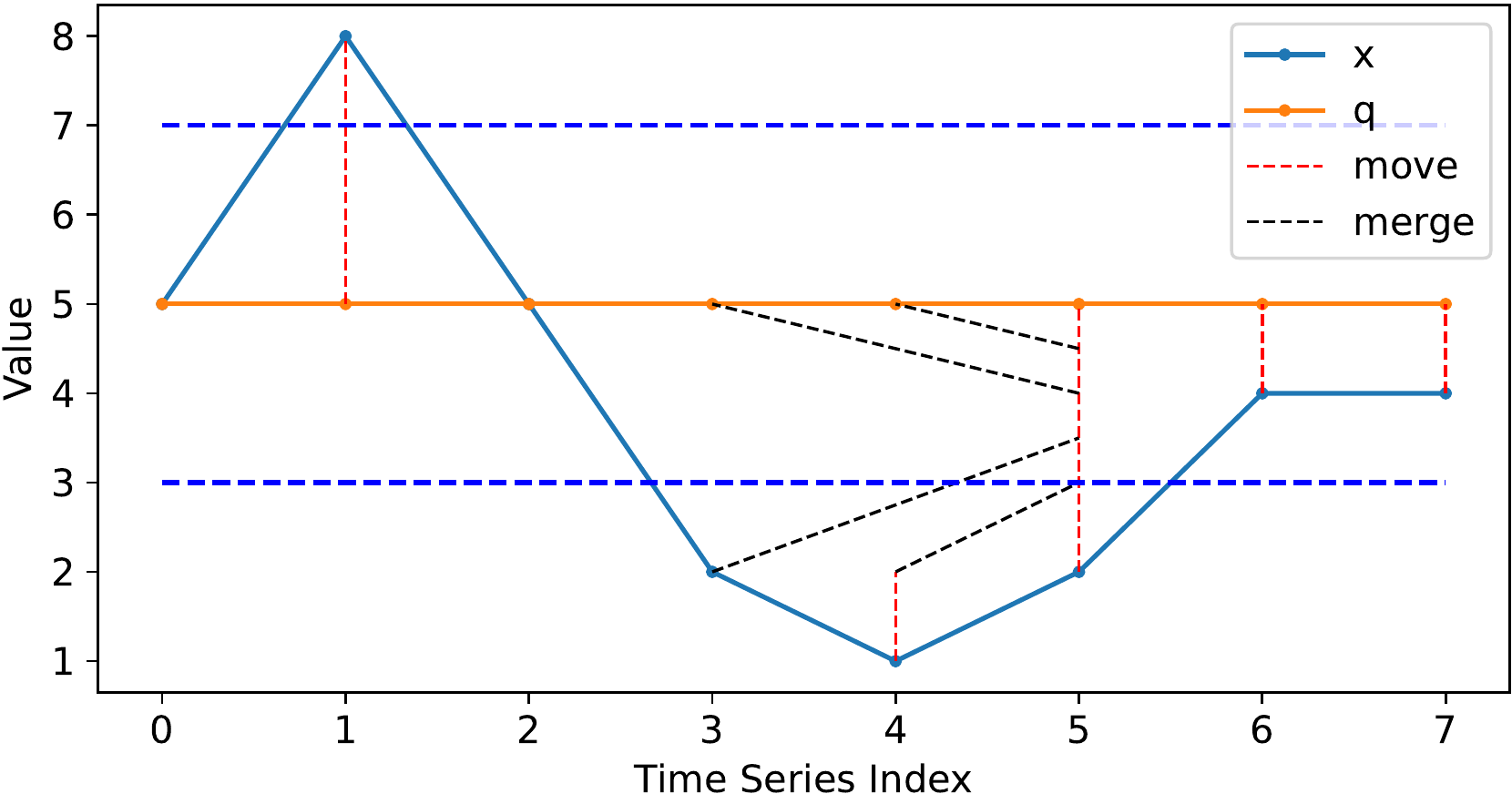}
    \caption{Alignment of $x=(5,8,5,2,1,2,4,4)$ to a constant time series with $q=5$. All merge operations are marked in red, merges and splits are marked in black. The two dotted lines refer to $c$ with $c=1$. The resulting dynamic programming table is $D_c=[13,13,10,10,8,5,2,1]$ with $d(x,q^{(8)})=13$.}
    \label{fig:constantDist2}
\end{figure}
In the following we give the computation rules of the dynamic program computing $d(x,q^{(n)})$. 
Lemma \ref{lemma:constantTsDecomposition} shows that we have to consider the cost of Type-1- and Type-4-Trees depending on the location of the points regarding the given thresholds $t^+ = q+2c$ and $t^- = q-2c$. Figure \ref{fig:constantDist2} depicts the logic of the cost computation. 

The base case is always a move operation, i.e., $D_c[n] = \vert x_n -q \vert$. 
For all other entries, we check if they are in a Type-1- or a Type-4-Tree. 
We can compute the cost per point. The cost for the first point in a Type-4-Tree are just the move costs. All further points in the same tree have cost that includes the constant cost $c$ for the merge and split operations and potentially some rest move cost to merge with another point. Formally, if $\vert x_n-q \vert \geq 2c$ and $\vert x_{n+1}-q \vert \geq 2c$ then
$$
D_c[i] = D_c[i+1] + \max(0, \vert x_{n+1} -x_n \vert),
$$
otherwise $D_c[i] = D_c[i+1] + \vert x_n - q \vert $.

\end{document}